\documentclass[journal]{IEEEtran}

\usepackage{amsmath,amssymb}
\usepackage{bm}
\usepackage{booktabs}
\usepackage{multirow}
\usepackage{graphicx}
\usepackage{cite}
\usepackage{url}
\usepackage{array}
\usepackage{tabularx}
\usepackage{siunitx}
\usepackage{setspace}
\usepackage{algorithm}
\usepackage{algpseudocode}
\newtheorem{theorem}{Theorem}
\newtheorem{proposition}{Proposition}
\newtheorem{assumption}{Assumption}

\DeclareMathOperator*{\argmin}{arg\,min}

\begin{document}
\setstretch{0.95}
\title{Frequency-Calibrated Decision-Focused Learning for BESS Dispatch Under Renewable Uncertainty}

\author{
Zihan~Meng,~\IEEEmembership{Student~Member,~IEEE,}
Chenxi~Hu,~\IEEEmembership{Member,~IEEE,}
Xin~Cui,~\IEEEmembership{Member,~IEEE,}
and~Francis~Yunhe~Hou,~\IEEEmembership{Fellow,~IEEE}%
\thanks{Z. Meng, X. Cui, and F. Y. Hou are with the Department of Electrical \& Computer Engineering, The University of Hong Kong, Hong Kong.}%
\thanks{C. Hu is with the Department of Computer Science, The Hang Seng University of Hong Kong, Hong Kong.}
}

\maketitle


\begin{abstract}

Battery energy storage system (BESS) flexibility depends on state of charge and intertemporal operating constraints. Renewable forecast errors of similar magnitudes can lead to different BESS dispatch results. An error may alter early charging or discharging decisions and reduce the flexibility available for subsequent regulation, while another may have little effect on the resulting schedule. To account for these unequal consequences, this paper proposes Frequency-Calibrated Decision-Focused Learning (FC-DFL), which trains renewable predictors using the realized operating costs of their induced BESS schedules. 
Gaussian smoothing and solver-based zeroth-order estimation provide training feedback without requiring differentiating through the non-differentiable mixed-integer BESS dispatch problem. The same decision feedback is further combined with prediction sensitivity to develop a decision-oriented XAI method that identifies the input features and forecast periods most consequential to BESS operation. Case studies using German wind and solar data show that FC-DFL reduces the integrated absolute frequency-deviation metric by 6.05\% and decision regret by 23.92\% relative to accuracy-based predict-then-optimize method. The XAI analysis further demonstrates that the identified high-impact forecast periods produce 2.84 times larger average absolute changes in realized operating cost than low-impact periods. These findings highlight the importance of incorporating downstream decision-making outcomes into renewable forecasting and demonstrate the value of FC-DFL in improving BESS regulation and providing operationally meaningful forecast explanations.

\end{abstract}

\begin{IEEEkeywords}
Battery energy storage system, decision-focused learning, explainable
artificial intelligence, frequency regulation, renewable forecasting,
renewable uncertainty.
\end{IEEEkeywords}


\section{Introduction}
\label{sec:introduction}

\IEEEPARstart{T}{he} increasing integration of wind and solar
generation introduces greater variability and uncertainty into
power-system operation. Renewable ramps change balancing
requirements over different time scales, and forecast errors
cause deviations between scheduled and realized renewable
generation
\cite{Ahmed2019ForecastingApps,Cui2017Ramping,Wang2018TemporalCorrelation}.
These errors may also be correlated across consecutive
intervals, affecting operating decisions over an extended
horizon. Accurate renewable forecasting is therefore
important for generation scheduling, reserve allocation,
and energy storage operation. However, forecast uncertainty
cannot be eliminated, and its operating consequences depend
on how the resulting predictions are used.

Renewable forecasting methods are commonly trained and
evaluated using statistical measures such as mean squared
error (MSE), root mean squared error (RMSE), and mean
absolute error (MAE). These metrics measure forecast
accuracy but do not directly quantify its value in
downstream operation
\cite{Ahmed2019ForecastingApps,Gandhi2024ForecastValue}.
Pinson et al.~\cite{Pinson2007WindTrading} demonstrated
how probabilistic wind power forecasts can improve trading
decisions by accounting for forecast uncertainty.
Xie et al.~\cite{Xie2014SpatioTemporalDispatch}
integrated spatiotemporal wind forecasts into robust
look-ahead dispatch and quantified the resulting operating
cost benefits. More broadly, Bertsimas and
Kallus~\cite{Bertsimas2020Prescriptive} developed a
data-driven framework for making decisions from contextual
information rather than treating prediction as an isolated
task. These studies establish the importance of evaluating
forecasts through their applications. For a constrained
energy storage resource, however, the value of a forecast
also depends on the physical operating conditions under
which its predicted trajectory is used.

Battery energy storage systems (BESSs) provide flexibility
through rapid charging and discharging and can support
power balancing and frequency regulation
\cite{Rouholamini2022BESSReview}.
Practical applications have demonstrated these capabilities.
For example, the Hornsdale Power Reserve has provided
frequency-control ancillary services in the Australian
electricity market \cite{AEMO2018Hornsdale}.
He et al.~\cite{He2017WindBatteryRegulation} developed
a coordinated wind--battery regulation strategy and
examined its participation in energy and regulation
markets. Xu et al.~\cite{Xu2018OptimalBatteryRegulation}
further investigated battery control and bidding for
frequency regulation, considering operating performance
and battery aging. These studies highlight the importance
of coordinating battery operation with regulation
requirements. The available BESS response also depends
on state of charge (SOC), charging and discharging power,
ramp limits, and the energy required in subsequent
intervals \cite{Liu2019RampESS}.

The same operating constraints also affect the
consequences of renewable forecast errors.
BESS dispatch is determined from a renewable forecast,
while its realized regulation performance depends on
the renewable generation that actually occurs.
Forecasts with similar error magnitudes can induce
different charging and discharging schedules and
SOC trajectories, leading to different realized costs.
An earlier BESS decision changes the stored energy
available in subsequent intervals, while power and
ramp limits further restrict the feasible dispatch
trajectory. Pointwise forecasting losses do not capture
these intertemporal operating consequences, motivating
the use of realized decision cost in predictor training.

Decision-focused learning (DFL) offers a means of
aligning prediction with optimization.
Task-based learning incorporates downstream objectives
into model training \cite{Donti2017TaskBased},
and subsequent approaches have developed methods
for learning from optimization decisions
\cite{Wilder2019DFL}.
Smart Predict-then-Optimize (SPO) formalizes
decision-based prediction losses for problems
with uncertain objective coefficients
\cite{Elmachtoub2022SPO}.
Related developments have also addressed uncertainty
beyond the objective. Estes and Richard
\cite{Estes2023RHS} studied SPO for two-stage
linear programs with uncertain right-hand-side
parameters, and Mandi et al.~\cite{Mandi2025ConstraintDFL}
examined feasibility and decision quality when
optimization constraints contain predicted parameters.
These advances demonstrate that DFL can address
different forms of optimization uncertainty.
Its application nevertheless requires a training
signal consistent with the physical role of the
predicted quantities and the downstream optimizer.

Recent research has also incorporated DFL into
energy storage applications. Yi et al.
\cite{Yi2025PerturbedStorageDFL} developed a
perturbation-based DFL framework for linear
storage models, with applications to energy
arbitrage and storage behavior prediction.
Kim et al.~\cite{Kim2025NeuralOptimizer}
proposed a neural optimizer that integrates
forecasting and BESS control for electricity
cost and peak-demand management.
These studies demonstrate the potential of
decision-focused methods in storage operation.
For BESS regulation under renewable uncertainty,
renewable generation enters the native power-balance
constraint as a physical injection, while BESS
dispatch includes binary charging/discharging
decisions and intertemporal SOC constraints.
The canonical SPO+ surrogate assumes uncertain
objective coefficients over a fixed feasible set
and does not directly accommodate renewable
generation entering the native power-balance
right-hand side. Furthermore, gradients obtained
through a continuous storage relaxation do not
directly capture the exact mixed-integer BESS
decisions. The resulting training signal must
therefore reflect the realized cost of forecast-induced
BESS dispatch without requiring derivatives
of the mixed-integer optimizer.

The distinction between prediction sensitivity
and realized BESS operating cost is also important
for model interpretability.
Explainable artificial intelligence (XAI) has
been increasingly studied in power and energy
systems to improve the transparency of
machine-learning models \cite{Machlev2022XAI}.
For solar forecasting, Mitrentsis and Lens
\cite{Mitrentsis2022SolarXAI} used SHAP to
interpret point forecasts and prediction
intervals. Liao et al.~\cite{Liao2024XAIWind}
investigated several XAI techniques for wind
power forecasting and evaluated the reliability
of their feature explanations. These approaches
provide insight into how predictor inputs
affect forecast outputs. For BESS regulation,
however, a feature with high prediction
importance does not necessarily have a large
effect on realized operating cost.
Conversely, a smaller forecast change can
be important when it alters BESS decisions
across constrained operating intervals.
This distinction calls for explanations
that connect prediction sensitivity with
downstream BESS operation.

To address these challenges, this paper proposes
Frequency-Calibrated Decision-Focused Learning
(FC-DFL) for renewable forecasting and BESS
regulation. FC-DFL uses the realized cost of
forecast-induced BESS dispatch to connect
predictor training with decision-oriented
interpretation. The main contributions of
FC-DFL are summarized as follows.

\begin{enumerate}

\item
We develop a decision-focused renewable forecasting
model for BESS regulation. FC-DFL evaluates the BESS
schedules induced by renewable forecasts under
realized renewable generation. The resulting
decision cost captures the effects of forecast
errors on intertemporal BESS operation and the
regulation capability available in subsequent
intervals.

\item
We develop an RHS-compatible decision-focused
training method for mixed-integer BESS dispatch.
FC-DFL employs Gaussian smoothing and solver-based
zeroth-order estimation to obtain decision gradients
from realized operating costs. The smoothed-gradient
identity and conditional unbiasedness of the
estimators are established under explicit assumptions.
The resulting feedback enables predictor training
without differentiating through the BESS optimizer.

\item
We develop a decision-oriented XAI method using
the same time-resolved decision gradient as
FC-DFL training. By combining this feedback with
the predictor Jacobian, the method identifies
input features and forecast periods with greater
importance to realized BESS operation.
Matched perturbation tests using the exact
BESS dispatch model validate the operational
relevance of the importance rankings.

\end{enumerate}

The remainder of this paper is organized as follows.
Section~\ref{sec:physical_model} develops the physical
BESS regulation model and realized decision cost.
Section~\ref{sec:fcdfl} presents FC-DFL training.
Section~\ref{sec:xai} introduces decision-oriented
explainability.
Section~\ref{sec:case_study} reports the case studies.
Section~\ref{sec:conclusion} concludes the paper.

\section{Forecast-Coupled BESS Regulation Model}
\label{sec:physical_model}

Renewable forecast errors affect BESS regulation through the
operating decisions they induce. This section first describes
renewable variability and forecast uncertainty, then formulates
the intertemporal BESS constraints. The resulting model establishes
how renewable forecasts determine BESS dispatch and its realized
operating cost. The overall framework of FC-DFL is illustrated
in Fig.~\ref{fig:overall_framework}.

\begin{figure}[t]
\centering
\includegraphics[width=\columnwidth]
{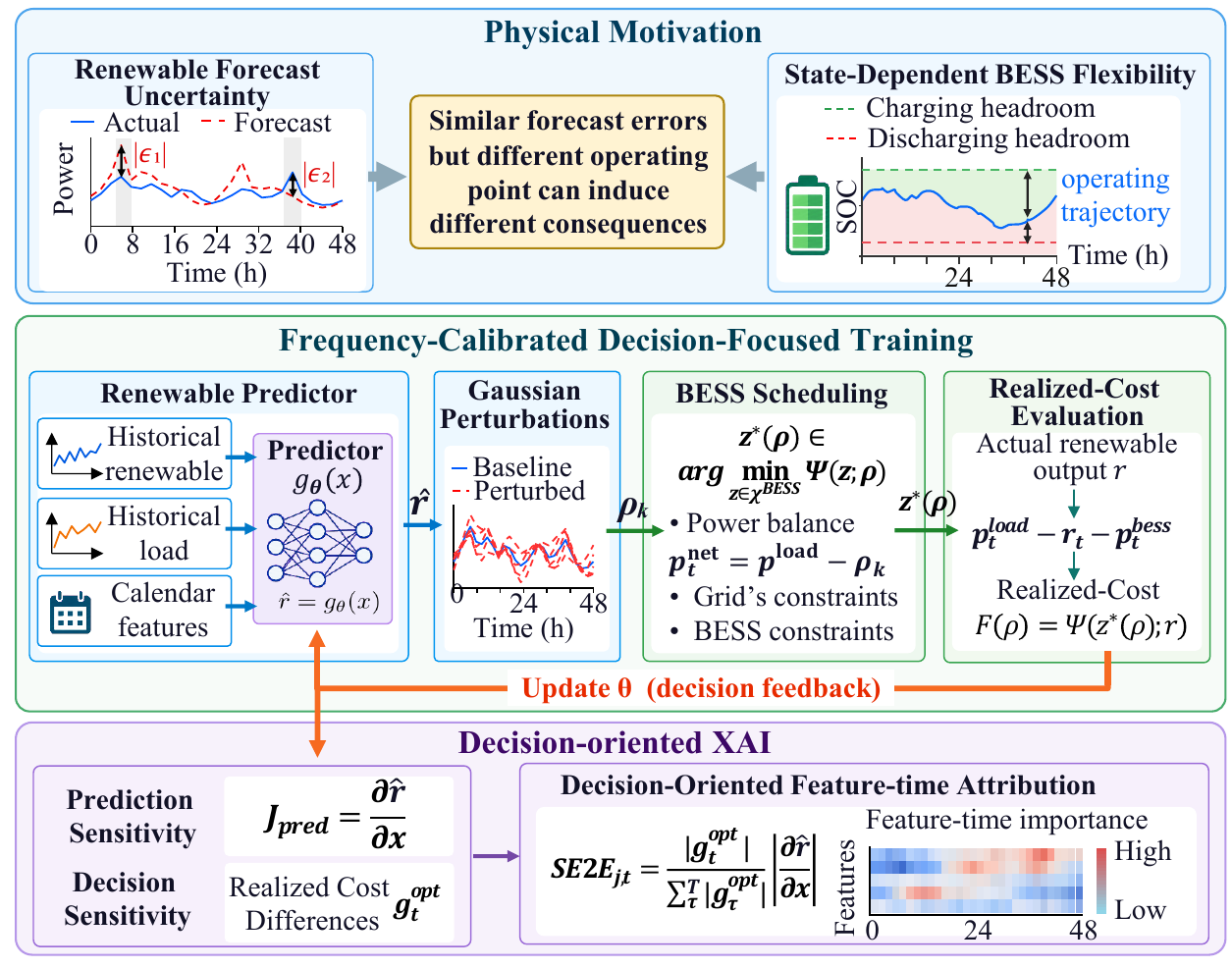}
\caption{Framework of FC-DFL for
renewable forecasting, BESS regulation,
and decision-oriented explainability.}
\label{fig:overall_framework}
\end{figure}


\subsection{Renewable Variability and Forecast Errors}
\label{subsec:renewable_structure}

Let $\mathcal T=\{1,\ldots,T\}$ denote the dispatch
intervals, each of duration $\Delta t$.
Wind and photovoltaic (PV) generation exhibit variations
over different time scales. Renewable ramps are commonly
characterized by their magnitude, rate, duration, and direction
\cite{GallegoCastillo2015RampReview,Florita2013Ramping,Cui2017Ramping}.
For sample $i$, the aggregate renewable generation and its
variation over $\tau$ intervals are expressed as
\begin{subequations}
\label{eq:renewable_structure}
\begin{align}
r_{i,t}
&=
r_{i,t}^{\mathrm{w}}
+
r_{i,t}^{\mathrm{pv}},
\qquad t\in\mathcal T,
\label{eq:renewable_total}
\\
\Delta_{\tau}r_{i,t}
&=
r_{i,t}-r_{i,t-\tau},
\qquad t=\tau+1,\ldots,T.
\label{eq:renewable_difference}
\end{align}
\end{subequations}

The renewable change $\Delta_{\tau}r_{i,t}$ captures both
the magnitude and direction of a ramp over $\tau\Delta t$.
Its average ramp rate follows from the change per unit time.
These characteristics describe the renewable power trajectory,
but they do not determine the regulation effort required from
the BESS. That effort also depends on the available charging
and discharging capability when the variation occurs.

Forecast uncertainty introduces a further difference between
scheduled and realized renewable generation. Following the
forecast-minus-realization convention
\cite{Zhang2018ForecastError}, the forecast errors and
conventional squared-error loss are given by
\begin{subequations}
\label{eq:forecast_error_structure}
\begin{align}
\epsilon_{i,t}
&=
\hat r_{i,t}-r_{i,t},
\label{eq:forecast_error}
\\
\bm\epsilon_i
&=
\hat{\bm r}_i-\bm r_i,
\label{eq:forecast_error_vector}
\\
\epsilon_{i,t,\tau}^{\mathrm{ramp}}
&=
\Delta_{\tau}\hat r_{i,t}
-
\Delta_{\tau}r_{i,t}
=
\epsilon_{i,t}-\epsilon_{i,t-\tau},
\label{eq:ramp_forecast_error}
\\
\mathcal L_i^{\mathrm{MSE}}
&=
\frac{1}{T}
\left\|
\bm\epsilon_i
\right\|_2^2.
\label{eq:mse}
\end{align}
\end{subequations}

The error $\epsilon_{i,t}$ measures the renewable mismatch
at one interval, whereas
$\epsilon_{i,t,\tau}^{\mathrm{ramp}}$ measures the error
in renewable change across two intervals.
MSE aggregates squared forecast errors without accounting
for their temporal ordering. Consequently, forecasts with
the same MSE can distribute their errors differently
across the dispatch horizon. The resulting operating
consequences depend on the BESS state and constraints,
as developed next.


\subsection{BESS Regulation Model}
\label{subsec:bess_model}

BESS regulation is constrained by the available converter
power, stored energy, and intertemporal operating limits.
Let $z_i$ collect the charging and discharging powers,
net BESS power, binary operating modes, and SOC trajectory
over the $T$-interval horizon. Positive
$p_{i,t}^{\mathrm{bess}}$ denotes net discharge to the system.
The rated power $\bar P$ is reduced by the retained
headroom $P^{\mathrm{res}}$, leaving $P^{\mathrm{use}}$
available for the modeled dispatch.

The BESS operating constraints are 
\begin{subequations}
\label{eq:bess_model}
\begin{align}
p_{i,t}^{\mathrm{bess}}
&=
p_{i,t}^{\mathrm{dch}}
-
p_{i,t}^{\mathrm{ch}},
\label{eq:bess_net_power}
\\
P^{\mathrm{use}}
&=
\bar P-P^{\mathrm{res}},
\label{eq:usable_power}
\\
0
&\le
p_{i,t}^{\mathrm{ch}}
\le
u_{i,t}P^{\mathrm{use}},
\label{eq:charge_power}
\\
0
&\le
p_{i,t}^{\mathrm{dch}}
\le
(1-u_{i,t})P^{\mathrm{use}},
\label{eq:discharge_power}
\\
u_{i,t}
&\in
\{0,1\},
\label{eq:bess_mode}
\\
s_{i,t+1}
&=
s_{i,t}
+
\frac{\Delta t}{\bar E}
\left(
\eta^{\mathrm{ch}}p_{i,t}^{\mathrm{ch}}
-
\frac{p_{i,t}^{\mathrm{dch}}}
{\eta^{\mathrm{dch}}}
\right),
\label{eq:energy_dynamics}
\\
\underline s
&\le
s_{i,t}
\le
\bar s,
\quad t\in\mathcal T\cup\{T+1\},
\label{eq:energy_soc_limits}
\\
s_{i,1}
&=
s_{i,T+1}
=
s^{\mathrm{init}},
\label{eq:terminal_energy}
\\
p_{i,0}^{\mathrm{bess}}
&=
p_{i,T+1}^{\mathrm{bess}}
=
0,
\label{eq:boundary_power}
\\
\left|
p_{i,t}^{\mathrm{bess}}
-
p_{i,t-1}^{\mathrm{bess}}
\right|
&\le
R^{\mathrm{bess}}\Delta t,
\quad t\in\mathcal T\cup\{T+1\}.
\label{eq:bess_ramp}
\end{align}
\end{subequations}

The binary operating mode enforces mutually exclusive
charging and discharging. The SOC dynamics in
\eqref{eq:energy_dynamics} account for conversion losses
and couple BESS decisions across consecutive intervals.
The terminal SOC condition prevents artificial energy
depletion over the scheduling horizon, while
\eqref{eq:boundary_power} and \eqref{eq:bess_ramp}
constrain power changes at both scheduling boundaries.

These constraints make the available BESS response
dependent on its operating trajectory. A high SOC provides
more energy for discharging but less capacity to absorb
renewable surplus, while a low SOC has the opposite effect.
Previous BESS power and the remaining scheduling horizon
further restrict admissible actions. Thus, similar renewable
variations can require different corrective schedules
depending on when they occur. The complete set of physical trajectories satisfying
\eqref{eq:bess_model} is denoted by
$\mathcal X_i^{\mathrm{BESS}}$.
This set includes the binary operating modes, SOC dynamics,
power limits, ramp constraints, and terminal condition.
It is fixed with respect to the renewable forecast.
Forecast changes affect the selected BESS trajectory,
rather than the underlying physical constraints.


\subsection{Forecast-Coupled Dispatch and Realized Decision Cost}
\label{subsec:forecast_coupled_model}

Renewable forecasts determine BESS dispatch through the
power-balance requirement. Let $p_{i,t}^{\mathrm{load}}$
denote the exogenous balance requirement presented to
the renewable--BESS regulation layer after accounting
for scheduled non-BESS supply. It remains fixed between
forecast-driven dispatch and realized evaluation, so
renewable generation is the only uncertain quantity
considered in the regulation problem.

For a generic renewable trajectory
$\bm\rho=[\rho_1,\ldots,\rho_T]^\top$,
the residual imbalance and forecast-coupled BESS
optimization problem are formulated as
\begin{subequations}
\label{eq:reduced_bess_problem}
\begin{align}
p_{i,t}^{\mathrm{net}}(\bm\rho)
&=
p_{i,t}^{\mathrm{load}}-\rho_t,
\label{eq:net_power}
\\
e_{i,t}(z_i;\bm\rho)
&=
p_{i,t}^{\mathrm{net}}(\bm\rho)
-
p_{i,t}^{\mathrm{bess}},
\label{eq:residual}
\\
\Delta f_{i,t}
&=
-\frac{
e_{i,t}(z_i;\bm\rho)
}{
B^{\mathrm{sys}}
},
\label{eq:frequency_proxy}
\\
\kappa
&=
1+
\frac{\lambda_f}{B^{\mathrm{sys}}},
\label{eq:kappa}
\\
\Psi_i(z_i;\bm\rho)
&=
\Delta t
\sum_{t=1}^{T}
\Big[
\kappa
\left|
e_{i,t}(z_i;\bm\rho)
\right|
\nonumber\\
&\qquad+
\lambda_p
\left(
p_{i,t}^{\mathrm{ch}}
+
p_{i,t}^{\mathrm{dch}}
\right)
\Big],
\label{eq:reduced_objective}
\\
z_i^\star(\bm\rho)
&\in
\argmin_{
z_i\in\mathcal X_i^{\mathrm{BESS}}
}
\Psi_i(z_i;\bm\rho).
\label{eq:reduced_optimizer}
\end{align}
\end{subequations}

The residual in \eqref{eq:residual} represents
the power imbalance remaining after renewable
generation and BESS regulation. A positive
residual indicates a power deficit and therefore
corresponds to a negative frequency deviation.
The relation in \eqref{eq:frequency_proxy}
provides a quasi-steady approximation rather
than a model of transient frequency dynamics. The operating objective combines absolute residual
imbalance with a throughput penalty on BESS charging
and discharging. The residual and frequency-referenced
penalties are incorporated into $\kappa$, while
$\lambda_p$ regularizes BESS activity without
representing electrochemical degradation.
Minimizing \eqref{eq:reduced_objective} over
$\mathcal X_i^{\mathrm{BESS}}$ determines the physical
BESS schedule induced by the renewable trajectory.

In the native mixed-integer formulation, renewable
generation enters the right-hand side of the power-balance
equation. Introducing nonnegative deficit and surplus
variables gives the equivalent representation
\begin{subequations}
\label{eq:native_rhs_problem}
\begin{align}
p_{i,t}^{\mathrm{bess}}
+
\delta_{i,t}^{\mathrm{def}}
-
\delta_{i,t}^{\mathrm{sur}}
&=
p_{i,t}^{\mathrm{net}}(\bm\rho),
\label{eq:native_balance}
\\
\delta_{i,t}^{\mathrm{def}},
\delta_{i,t}^{\mathrm{sur}}
&\ge
0,
\label{eq:imbalance_variables}
\\
\min_{
z_i,\bm\delta_i^{\mathrm{def}},
\bm\delta_i^{\mathrm{sur}}
}
\quad
&
\Delta t
\sum_{t=1}^{T}
\Big[
\kappa
\left(
\delta_{i,t}^{\mathrm{def}}
+
\delta_{i,t}^{\mathrm{sur}}
\right)
\nonumber\\
&\qquad+
\lambda_p
\left(
p_{i,t}^{\mathrm{ch}}
+
p_{i,t}^{\mathrm{dch}}
\right)
\Big]
\label{eq:native_objective}
\\
\text{s.t.}\quad
&
z_i\in\mathcal X_i^{\mathrm{BESS}},
\quad
\eqref{eq:native_balance},
\quad
\eqref{eq:imbalance_variables}.
\nonumber
\end{align}
\end{subequations}

The deficit and surplus variables account for the
remaining power imbalance in opposite directions.
Unlike the reduced formulation, the native model
places the renewable trajectory explicitly in the
balance constraint. The two representations nevertheless
produce the same optimal physical BESS dispatch,
as established below.

\begin{proposition}
\label{prop:rhs_equivalence}
For any renewable trajectory $\bm\rho$ and fixed
physical BESS action $z_i\in\mathcal X_i^{\mathrm{BESS}}$,
the optimal deficit and surplus variables in
\eqref{eq:native_rhs_problem} satisfy
\begin{subequations}
\label{eq:optimal_slacks}
\begin{align}
\delta_{i,t}^{\mathrm{def}}
&=
\max
\left\{
e_{i,t}(z_i;\bm\rho),0
\right\},
\label{eq:deficit_opt}
\\
\delta_{i,t}^{\mathrm{sur}}
&=
\max
\left\{
-e_{i,t}(z_i;\bm\rho),0
\right\},
\label{eq:surplus_opt}
\\
\delta_{i,t}^{\mathrm{def}}
+
\delta_{i,t}^{\mathrm{sur}}
&=
\left|
e_{i,t}(z_i;\bm\rho)
\right|.
\label{eq:slack_abs}
\end{align}
\end{subequations}
Consequently, \eqref{eq:native_rhs_problem} and
\eqref{eq:reduced_optimizer} are equivalent.
\end{proposition}

\noindent\emph{Proof:}
For fixed $z_i$ and $\bm\rho$,
\eqref{eq:native_balance} requires
$\delta_{i,t}^{\mathrm{def}}-
\delta_{i,t}^{\mathrm{sur}}
=e_{i,t}(z_i;\bm\rho)$.
Since $\kappa>0$, minimizing the sum of the
nonnegative imbalance variables yields the
positive and negative parts of the residual,
respectively. Their sum equals its absolute
value, as given in \eqref{eq:optimal_slacks}.
Substituting this result into
\eqref{eq:native_objective} recovers
\eqref{eq:reduced_objective}, with the same
physical BESS feasible set.
\hfill$\square$

The operating consequence of a forecast is determined
by evaluating its induced BESS schedule under realized
renewable generation. The realized decision cost,
realization-aware reference, decision regret, and
frequency-deviation metric are defined by
\begin{subequations}
\label{eq:realized_decision_model}
\begin{align}
F_i(\bm\rho)
&=
\Psi_i
\left(
z_i^\star(\bm\rho);
\bm r_i
\right),
\label{eq:realized_cost}
\\
z_i^{\mathrm{ref}}
&=
z_i^\star(\bm r_i),
\label{eq:reference_action}
\\
F_i(\bm r_i)
&=
\min_{
z_i\in\mathcal X_i^{\mathrm{BESS}}
}
\Psi_i(z_i;\bm r_i),
\label{eq:reference_cost}
\\
\ell_i^{\mathrm{dec}}(\bm\rho)
&=
F_i(\bm\rho)
-
F_i(\bm r_i),
\label{eq:decision_regret}
\\
\mathrm{IAFD}_i(z_i)
&=
\frac{\Delta t}{B^{\mathrm{sys}}}
\sum_{t=1}^{T}
\left|
p_{i,t}^{\mathrm{load}}
-
r_{i,t}
-
p_{i,t}^{\mathrm{bess}}
\right|.
\label{eq:iafd}
\end{align}
\end{subequations}

The realized cost $F_i(\bm\rho)$ evaluates the
forecast-induced BESS dispatch under realized
renewable generation. The resulting imbalance
is recomputed using $\bm r_i$, while the
physical BESS dispatch remains unchanged. The realization-aware reference instead optimizes
directly with $\bm r_i$ under the same physical
constraints and exogenous balance requirement.
It therefore provides the minimum realized cost
for that sample. Since
$z_i^\star(\bm\rho)\in\mathcal X_i^{\mathrm{BESS}}$,
the decision regret in \eqref{eq:decision_regret}
is nonnegative. The IAFD in \eqref{eq:iafd}
separately quantifies the accumulated magnitude
of realized frequency deviations. For a fixed BESS action, the difference between
the realized and forecast residuals equals
$\epsilon_{i,t}$. However, changing the renewable
forecast also changes the optimized BESS schedule
over the entire horizon. Consequently, forecast
errors with similar magnitudes can induce different
charging and discharging decisions, SOC trajectories,
and realized costs. This distinction links renewable forecast uncertainty
to the intertemporal BESS constraints. Pointwise
forecast accuracy does not account for the operating
consequences of the induced schedule.
The FC-DFL method developed in Section~\ref{sec:fcdfl}
therefore trains the renewable predictor using
the realized decision cost in \eqref{eq:realized_cost}.

\section{Frequency-Calibrated Decision-Focused Learning}
\label{sec:fcdfl}

The BESS model in Section~\ref{sec:physical_model}
shows that renewable forecast errors affect realized
regulation performance through the dispatch they induce.
FC-DFL incorporates this operating consequence into
predictor training. Gaussian smoothing provides a differentiable surrogate of the realized decision cost, while zeroth-order
estimation obtains its gradient without differentiating
through the mixed-integer BESS optimizer.


\subsection{Gaussian-Smoothed Realized Decision Cost}
\label{subsec:gaussian_smoothing}

Unlike the canonical cost-vector SPO formulation
\cite{Elmachtoub2022SPO}, renewable generation
enters the BESS power-balance right-hand side.
The equivalent residual formulation in
Section~\ref{subsec:forecast_coupled_model}
allows forecast-induced dispatch to be evaluated
over the fixed physical BESS feasible set. The realized decision cost in \eqref{eq:realized_cost}
is evaluated using the exact forecast-induced BESS
schedule. Since the optimizer contains binary operating
decisions, its dependence on the renewable forecast
need not be differentiable. FC-DFL instead smooths the realized decision
cost evaluated at projected renewable forecasts.

Let $\mathcal R\subset\mathbb R^T$ be a nonempty,
compact, and convex set of admissible renewable
forecast trajectories, with $\Pi_{\mathcal R}$
denoting Euclidean projection onto this set.
For $\bm U\sim\mathcal N(0,I_T)$ and $\sigma>0$,
the predictor, projected cost, and smoothed
decision objective are defined by
\begin{subequations}
\label{eq:smoothed_objective}
\begin{align}
\hat{\bm r}_i
&=
g_{\theta}(\bm x_i)
\in\mathbb R^T,
\label{eq:renewable_predictor}
\\
\bar F_i(\bm y)
&=
F_i\left(\Pi_{\mathcal R}(\bm y)\right),
\quad \bm y\in\mathbb R^T,
\label{eq:Fbar}
\\
\widetilde F_{i,\sigma}(\bm\rho)
&=
\mathbb E_{\bm U}
\left[
\bar F_i(\bm\rho+\sigma\bm U)
\right],
\label{eq:smoothed_F}
\\
\widetilde\ell_{i,\sigma}^{\mathrm{dec}}(\bm\rho)
&=
\widetilde F_{i,\sigma}(\bm\rho)
-
F_i(\bm r_i).
\label{eq:smoothed_regret}
\end{align}
\end{subequations}

Projection keeps forecast perturbations within
$\mathcal R$, without altering the realized
renewable trajectory $\bm r_i$ used for
cost evaluation and the reference. Because this reference cost is independent of the
forecast, smoothing the decision regret gives the
same gradient as smoothing $F_i$. The following
assumptions establish the mathematical conditions
required for this gradient.

\begin{assumption}
\label{ass:optimizer}
For every $\bm\rho\in\mathcal R$, the BESS problem
\eqref{eq:reduced_optimizer} admits an optimal
physical action. A fixed Borel-measurable selection
rule determines $z_i^\star(\bm\rho)$ when multiple
optimal solutions exist.
\end{assumption}

\begin{assumption}
\label{ass:bounded_F}
The restriction of the realized decision-cost
function $F_i$ to $\mathcal R$ is bounded
and Borel measurable.
\end{assumption}

\begin{theorem}
\label{thm:smoothed_gradient}
Under Assumptions~\ref{ass:optimizer} and
\ref{ass:bounded_F}, the smoothed decision cost
$\widetilde F_{i,\sigma}$ is infinitely
differentiable on $\mathbb R^T$ for every
$\sigma>0$, with
\begin{align}
\nabla_{\bm\rho}
\widetilde F_{i,\sigma}(\bm\rho)
&=
\frac{1}{\sigma}
\mathbb E_{\bm U}
\left[
\bar F_i(\bm\rho+\sigma\bm U)\bm U
\right].
\label{eq:smoothed_gradient_identity}
\end{align}
\end{theorem}

\noindent\emph{Proof:}
The smoothed objective in \eqref{eq:smoothed_F}
is the convolution of the bounded measurable
function $\bar F_i$ with a Gaussian density.
Every derivative of the Gaussian kernel is
integrable, allowing differentiation under
the integral. Differentiating with respect to
$\bm\rho$ and expressing the result using
$\bm U\sim\mathcal N(0,I_T)$ gives
\eqref{eq:smoothed_gradient_identity}.
The same argument applies to derivatives
of higher orders.
\hfill$\square$

Theorem~\ref{thm:smoothed_gradient} provides a
decision gradient without requiring derivatives
of the mixed-integer optimizer.
Evaluated at the renewable forecast, this
time-resolved feedback is
\begin{align}
g_i^{\mathrm{opt}}
&=
\nabla_{\hat{\bm r}_i}
\widetilde F_{i,\sigma}
\left(\hat{\bm r}_i\right)
\in\mathbb R^T.
\label{eq:gopt}
\end{align}

The component $g_{i,t}^{\mathrm{opt}}$
measures the sensitivity of the smoothed realized
decision cost to the renewable forecast at interval
$t$. Unlike a fixed penalty on forecast error,
this sensitivity reflects the BESS schedule
obtained from the complete forecast trajectory
and the realized operating cost it produces. At continuity points of $\bar F_i$,
$\widetilde F_{i,\sigma}(\bm\rho)$ converges
to $\bar F_i(\bm\rho)$ as $\sigma\to0$.


\subsection{Zeroth-Order Decision Feedback}
\label{subsec:zeroth_order}

The gradient in
\eqref{eq:smoothed_gradient_identity}
can be estimated from realized decision-cost
evaluations obtained using the exact BESS
optimizer. For independent directions
$\bm u_{i,k}\sim\mathcal N(0,I_T)$,
$k=1,\ldots,K$, the one-sided and antithetic
estimators are
\begin{subequations}
\label{eq:gradient_estimators}
\begin{align}
\widehat g_{i,\mathrm{one}}^{\mathrm{opt}}
&=
\frac{1}{K\sigma}
\sum_{k=1}^{K}
\Big[
\bar F_i
\left(
\hat{\bm r}_i+\sigma\bm u_{i,k}
\right)
\nonumber\\
&\qquad-
\bar F_i
\left(\hat{\bm r}_i\right)
\Big]\bm u_{i,k},
\label{eq:one_sided_estimator}
\\
\widehat g_{i,\mathrm{anti}}^{\mathrm{opt}}
&=
\frac{1}{2K\sigma}
\sum_{k=1}^{K}
\Big[
\bar F_i
\left(
\hat{\bm r}_i+\sigma\bm u_{i,k}
\right)
\nonumber\\
&\qquad-
\bar F_i
\left(
\hat{\bm r}_i-\sigma\bm u_{i,k}
\right)
\Big]\bm u_{i,k}.
\label{eq:antithetic_estimator}
\end{align}
\end{subequations}

\begin{proposition}
\label{prop:estimator_unbiased}
Under the conditions of
Theorem~\ref{thm:smoothed_gradient},
both estimators in
\eqref{eq:gradient_estimators}
are conditionally unbiased for the same
smoothed decision gradient:
\begin{subequations}
\label{eq:estimator_unbiased}
\begin{align}
\mathbb E
\left[
\widehat g_{i,\mathrm{one}}^{\mathrm{opt}}
\mid\theta,\bm x_i,\bm r_i
\right]
&=
g_i^{\mathrm{opt}},
\label{eq:one_unbiased}
\\
\mathbb E
\left[
\widehat g_{i,\mathrm{anti}}^{\mathrm{opt}}
\mid\theta,\bm x_i,\bm r_i
\right]
&=
g_i^{\mathrm{opt}}.
\label{eq:anti_unbiased}
\end{align}
\end{subequations}
\end{proposition}

\noindent\emph{Proof:}
Conditional on the current sample and predictor
parameters, $\bar F_i(\hat{\bm r}_i)$ is fixed.
Since $\mathbb E[\bm U]=0$, subtracting this
baseline does not change the expected gradient.
Theorem~\ref{thm:smoothed_gradient} therefore
establishes \eqref{eq:one_unbiased}.
For the antithetic estimator, the symmetry
of the Gaussian distribution gives
\[
\mathbb E
\left[
\bar F_i(\hat{\bm r}_i-\sigma\bm U)\bm U
\right]
=
-
\mathbb E
\left[
\bar F_i(\hat{\bm r}_i+\sigma\bm U)\bm U
\right].
\]
Substitution into \eqref{eq:antithetic_estimator}
yields \eqref{eq:anti_unbiased}.
\hfill$\square$
Each perturbed renewable forecast induces a BESS
dispatch schedule whose realized cost reflects
intertemporal SOC constraints and operating decisions.
The resulting decision gradient provides feedback
on BESS regulation performance rather than
forecast accuracy alone. FC-DFL training uses the baseline-centered
one-sided estimator, requiring $K+1$ BESS
optimizations per gradient estimate.
The antithetic estimator estimates the same
gradient and is used for decision-oriented
analysis in Section~\ref{sec:xai}.


\subsection{FC-DFL Predictor Training}
\label{subsec:fcdfl_training}

The smoothed decision objective connects
renewable forecasting with realized BESS
regulation performance. For $N$ training
samples, FC-DFL minimizes the average
smoothed decision regret. The corresponding
gradient and backpropagation surrogate are
\begin{subequations}
\label{eq:fcdfl_training}
\begin{align}
\mathcal L_{\mathrm{FC}}(\theta)
&=
\frac{1}{N}
\sum_{i=1}^{N}
\widetilde\ell_{i,\sigma}^{\mathrm{dec}}
\left(
g_{\theta}(\bm x_i)
\right),
\label{eq:fc_objective}
\\
\nabla_{\theta}
\mathcal L_{\mathrm{FC}}(\theta)
&=
\frac{1}{N}
\sum_{i=1}^{N}
\left[
\nabla_{\theta}
g_{\theta}(\bm x_i)
\right]^\top
g_i^{\mathrm{opt}},
\label{eq:fc_chain_rule}
\\
\mathcal L_i^{\mathrm{back}}
&=
\left\langle
\operatorname{sg}
\left(
\widehat g_{i,\mathrm{one}}^{\mathrm{opt}}
\right),
\hat{\bm r}_i
\right\rangle,
\label{eq:pseudo_loss}
\\
\nabla_{\theta}
\mathcal L_i^{\mathrm{back}}
&=
\left[
\nabla_{\theta}
g_{\theta}(\bm x_i)
\right]^\top
\widehat g_{i,\mathrm{one}}^{\mathrm{opt}}.
\label{eq:pseudo_gradient}
\end{align}
\end{subequations}

The predictor Jacobian combines with the
decision gradient $g_i^{\mathrm{opt}}$,
allowing forecast periods to contribute
according to their realized operating
consequences. In implementation, the stochastic feedback
$\widehat g_{i,\mathrm{one}}^{\mathrm{opt}}$
replaces the expected gradient in this
expression. The stop-gradient operator
$\operatorname{sg}(\cdot)$ in
\eqref{eq:pseudo_loss} treats the solver
feedback as fixed during backpropagation.
Differentiating the surrogate therefore
passes the estimated decision gradient
through the renewable predictor without
differentiating the mixed-integer BESS
optimization problem. The resulting training procedure adjusts
renewable forecasts according to the
realized costs of their induced BESS
schedules, rather than minimizing forecast
MSE alone. Because the same
$g_i^{\mathrm{opt}}$ is defined for each
forecast interval, it also provides the
optimization feedback used for
decision-oriented interpretation in
Section~\ref{sec:xai}.

\section{Decision-Oriented Explainability}
\label{sec:xai}

The optimization feedback developed in
Section~\ref{sec:fcdfl} links renewable forecast
changes to realized BESS operating cost.
This section uses the same feedback to identify
predictor inputs whose effects on renewable forecasts
have greater downstream importance.
The resulting importance scores are evaluated through
matched perturbations and exact BESS reoptimization.


\subsection{Decision-Oriented Feature Importance}
\label{subsec:exact_xai}

Prediction sensitivity measures how an input feature
affects renewable forecasts, but it does not account
for the resulting BESS operating decisions.
To incorporate the downstream consequence,
consider the differentiable predictor
$g_\theta$ and the smoothed realized cost
$\widetilde F_{i,\sigma}$.
Their relationship is given by
\begin{subequations}
\label{eq:exact_xai}
\begin{align}
J_i^{\mathrm{pred}}
&=
\nabla_{\bm x_i}g_\theta(\bm x_i)
\in\mathbb R^{T\times d},
\label{eq:prediction_jacobian}
\\
G_{i,j,t}^{\mathrm{pred}}
&=
\left|
\left[J_i^{\mathrm{pred}}\right]_{t,j}
\right|,
\label{eq:prediction_magnitude}
\\
H_i(\bm x)
&=
\widetilde F_{i,\sigma}
\left(g_\theta(\bm x)\right),
\label{eq:downstream_map}
\\
\frac{\partial H_i(\bm x_i)}
{\partial x_{i,j}}
&=
\sum_{t\in\mathcal T}
g_{i,t}^{\mathrm{opt}}
\left[J_i^{\mathrm{pred}}\right]_{t,j}.
\label{eq:signed_sum}
\end{align}
\end{subequations}

The predictor Jacobian captures the sensitivity
of renewable forecasts to input features.
In \eqref{eq:signed_sum}, this sensitivity is
weighted by the decision gradient, which reflects
the realized cost of forecast-induced BESS dispatch.
A feature may therefore strongly influence renewable
predictions yet have limited importance for BESS
operation. For importance ranking, the decision-gradient
magnitude is normalized across the forecast horizon
and combined with prediction sensitivity:
\begin{subequations}
\label{eq:xai_magnitude_score}
\begin{align}
\bar g_{i,t}^{\mathrm{opt}}
&=
\begin{cases}
\dfrac{|g_{i,t}^{\mathrm{opt}}|}
{\displaystyle\sum_{\tau\in\mathcal T}
|g_{i,\tau}^{\mathrm{opt}}|},
&
\|g_i^{\mathrm{opt}}\|_1>0,
\\[2mm]
0,
&
\|g_i^{\mathrm{opt}}\|_1=0,
\end{cases}
\label{eq:normalized_gopt}
\\
S_{i,j,t}^{\mathrm{E2E}}
&=
\bar g_{i,t}^{\mathrm{opt}}
G_{i,j,t}^{\mathrm{pred}},
\label{eq:e2e_score}
\\
S_{i,j}^{\mathrm{feat}}
&=
\sum_{t\in\mathcal T}
S_{i,j,t}^{\mathrm{E2E}},
\label{eq:feature_score}
\\
S_{i,t}^{\mathrm{time}}
&=
\bar g_{i,t}^{\mathrm{opt}}.
\label{eq:time_score}
\end{align}
\end{subequations}

The feature score in \eqref{eq:feature_score}
aggregates prediction sensitivity weighted by
the normalized decision gradient across the
forecast horizon. The temporal score in
\eqref{eq:time_score} ranks forecast periods
according to their normalized decision feedback.
Therefore, high prediction sensitivity does not
necessarily imply high decision-oriented importance,
as the latter also depends on the realized-cost
sensitivity of the affected forecast periods.
These magnitude-based scores are used for
importance ranking.


\subsection{Estimation of Decision-Oriented Importance Scores}
\label{subsec:xai_mc}

The decision gradient is estimated using
the antithetic estimator in
\eqref{eq:antithetic_estimator}.
This estimator evaluates the realized BESS cost
under symmetric renewable forecast perturbations
and targets the same smoothed gradient used
in FC-DFL training.

For $K_{\mathrm{xai}}$ antithetic perturbation
pairs, the estimated feedback and importance
scores are
\begin{subequations}
\label{eq:xai_mc}
\begin{align}
\widehat g_i^{\mathrm{xai}}
&=
\left.
\widehat g_{i,\mathrm{anti}}^{\mathrm{opt}}
\right|_{K=K_{\mathrm{xai}}},
\label{eq:xai_gradient_average}
\\
\widehat{\bar g}_{i,t}^{\mathrm{opt}}
&=
\begin{cases}
\dfrac{
|\widehat g_{i,t}^{\mathrm{xai}}|
}{
\displaystyle\sum_{\tau\in\mathcal T}
|\widehat g_{i,\tau}^{\mathrm{xai}}|
},
&
\|\widehat g_i^{\mathrm{xai}}\|_1>0,
\\[2mm]
0,
&
\|\widehat g_i^{\mathrm{xai}}\|_1=0,
\end{cases}
\label{eq:estimated_normalized_gradient}
\\
\widehat S_{i,j,t}^{\mathrm{E2E}}
&=
\widehat{\bar g}_{i,t}^{\mathrm{opt}}
G_{i,j,t}^{\mathrm{pred}}.
\label{eq:estimated_xai_score}
\end{align}
\end{subequations}

The estimated decision gradient captures the
temporal sensitivity of realized BESS cost to
renewable forecast changes. By weighting the
predictor sensitivity with this feedback,
$\widehat S_{i,j,t}^{\mathrm{E2E}}$ identifies
input features that influence forecast periods
with greater importance to BESS operation.
The resulting rankings are further examined
through BESS reoptimization in
Section~\ref{subsec:faithfulness}.


\subsection{Operational Faithfulness}
\label{subsec:faithfulness}

The importance ranking is validated by
comparing the realized BESS cost changes
induced by matched perturbations of
high- and low-ranked groups.
All perturbed forecasts are evaluated
using the same physical BESS model and
realized renewable generation. For feature-level validation, the downstream
cost associated with a modified predictor
input is
\begin{subequations}
\label{eq:feature_faithfulness}
\begin{gather}
Q_i(\bm x)
=
F_i\left(
\Pi_{\mathcal R}\left[g_\theta(\bm x)\right]
\right),
\label{eq:exact_downstream_map}
\\
\begin{aligned}[b]
\Delta_{i,A}^{\mathrm{feat}}
=
\frac{1}{2}
\sum_{\nu\in\{-1,+1\}}
\Big|
Q_i\left(
\bm x_i+\nu\delta_x
\sum_{j\in A}\bm e_j^{(d)}
\right) -
Q_i(\bm x_i)
\Big|.
\end{aligned}
\label{eq:feature_perturbation_effect}
\end{gather}
\end{subequations}

High- and low-ranked renewable-history feature
groups are selected by aggregating
$\widehat S_{i,j,t}^{\mathrm{E2E}}$ over the
forecast horizon.
All features within a group receive the same
signed perturbation simultaneously.
Only features permitting symmetric changes
without negative historical generation are eligible.
The remaining contextual inputs and exogenous
regulation requirement are held fixed.

Forecast periods are ranked using the
estimated temporal decision importance
$\widehat{\bar g}_{i,t}^{\mathrm{opt}}$:
\begin{subequations}
\label{eq:temporal_faithfulness}
\begin{gather}
\bm\rho_i^0
=
\Pi_{\mathcal R}\left(\hat{\bm r}_i\right),
\label{eq:forecast_center}
\\
\begin{aligned}[b]
\Delta_{i,A}^{\mathrm{time}}
=
\frac{1}{2} 
\sum_{\nu\in\{-1,+1\}}
\Big|
F_i\left(
\bm\rho_i^0+\nu\delta_{r,i}
\sum_{t\in A}\bm e_t^{(T)}
\right)
-F_i(\bm\rho_i^0)
\Big|.
\end{aligned}
\label{eq:time_perturbation_effect}
\end{gather}
\end{subequations}

Equal-sized high- and low-ranked groups receive
matched positive and negative perturbations.
For forecast-period groups, the common step
$\delta_{r,i}$ is selected so that both
perturbation directions remain in $\mathcal R$.
The resulting changes in realized decision cost
are evaluated through BESS reoptimization. Comparing the cost changes of high- and low-ranked
groups provides an empirical assessment of
whether the importance scores reflect downstream
BESS operation. The FC-DFL training and decision-oriented
explanation procedure is summarized in
Algorithm~\ref{alg:fcdfl_xai}.
The method uses the same realized BESS cost
for both stages, but different estimators
of the corresponding smoothed decision gradient.

\begin{algorithm}[t]
\caption{\textbf{FC-DFL Training and Decision-Oriented Explainability}}
\label{alg:fcdfl_xai}
\begin{algorithmic}[1]

\Require
Training and validation data;
XAI samples;
BESS model;
$\mathcal R$, $\sigma$, $K$, and $K_{\mathrm{xai}}$.

\Ensure
Trained predictor $\theta^\star$,
decision-oriented importance scores,
and perturbation results.

\State Initialize the renewable predictor $g_\theta$.

\Repeat
    \State Evaluate baseline and perturbed realized
    BESS costs for each sample in a training minibatch.

    \State Estimate $\widehat g_{i,\mathrm{one}}^{\mathrm{opt}}$
    using \eqref{eq:one_sided_estimator}
    and update $\theta$ through \eqref{eq:pseudo_loss}.

\Until{the training stopping criterion is satisfied}

\State Select $\theta^\star$ according to
realized validation cost.

\For{each XAI sample $i$}
    \State Compute prediction sensitivity and
    antithetic decision feedback.

    \State Rank features and forecast periods
    using \eqref{eq:e2e_score}--\eqref{eq:time_score}.

    \State Evaluate matched group perturbations
    through BESS reoptimization.
\EndFor

\end{algorithmic}
\end{algorithm}

\section{Case Studies}
\label{sec:case_study}

This section evaluates FC-DFL in terms of realized BESS regulation
performance, physical operating behavior, robustness, and
decision-oriented interpretability.

\subsection{Experimental Settings}
\label{subsec:experimental_settings}

The experiments use German hourly time-series data from
Open Power System Data (OPSD)~\cite{OPSD2020},
including measured wind and solar generation and system load. Each sample consists of a 48-h observation history followed by
a 48-h renewable forecasting and BESS dispatch horizon.
The predictor uses historical renewable generation, historical
load, and cyclical calendar features.
The evaluation includes 560 training, 140 validation, and
100 chronologically held-out test windows with non-overlapping
dispatch horizons. The test set covers March--September 2020,
providing 4,800 h of measured renewable generation for
out-of-sample evaluation. The exogenous balance requirement
$p_{i,t}^{\mathrm{load}}$ is constructed by repeating
the preceding 24-h renewable profile over the dispatch
horizon. Scheduled non-BESS supply is determined
from the measured system load and this requirement.
All methods face the same balance requirement,
with renewable generation as the only uncertain quantity.
A common scaling factor determined from the training
data maps the measured trajectories to a study system
with an average system load of 80 MW. The nominal BESS is rated at 100 MW/400 MWh, with 30 MW
of retained power headroom. Charging and discharging
efficiencies are 0.92, and SOC is constrained to $[0.2,0.9]$
with initial and terminal values of 0.5. The ramp-rate limit
is 60 MW/h, and the frequency coefficient is
$B^{\mathrm{sys}}=5000$ MW/Hz. The objective weights are
$\lambda_f=0.1$ and $\lambda_p=10^{-4}$.

Six methods are compared under identical BESS operating
conditions. Uncontrolled applies no BESS regulation, while
the Realization-Aware Reference uses realized renewable
generation. ERM-PTO minimizes forecast MSE, DS-DFL employs
a differentiable BESS surrogate, and FC-DFL uses the proposed
Gaussian-smoothed decision feedback. Empirical Scenario SD
optimizes dispatch using ERM-PTO forecasts and ten empirical
validation-error scenarios. The three learned predictors
share the same architecture, with DS-DFL and FC-DFL
initialized from ERM-PTO and selected using realized
validation cost. FC-DFL uses $\sigma_0=1.3483$ MW and $K=4$.
All methods are evaluated over five independent random seeds
and 100 common test windows using the same binary BESS model. The computational framework is implemented in Python using
PyTorch 2.14.1 for learning, CVXPY 1.6.7 and CVXPYLayers 0.1.9
for differentiable surrogate training, and Gurobi 13.0.3
for mixed-integer BESS dispatch optimization.

\subsection{Forecast Accuracy and Realized Regulation Performance}
\label{subsec:main_performance}

\begin{table}[t]
\centering
\caption{Forecast accuracy and realized BESS regulation performance.}
\label{tab:main_performance}
\setlength{\tabcolsep}{2pt}
\begin{tabular}{lccccc}
\hline \hline
Method
&
\shortstack{RMSE\\(MW)}
&
\shortstack{IAFD\\(Hz$\cdot$h)}
&
\shortstack{IAFD\\reduction\\(\%)}
&
\shortstack{Decision\\regret}
&
\shortstack{BESS\\throughput\\(MWh)}
\\
\hline

Uncontrolled
& --
& 0.099
& 0.00
& 114.70
& 0.0
\\

ERM-PTO
& $\mathbf{9.73}$
& 0.101
& -2.77
& 128.37
& 204.0
\\

DS-DFL
& 9.76
& 0.100
& -1.82
& 123.68
& 194.8
\\

Scenario SD
& $\mathbf{9.73}$
& 0.096
& 2.91
& 100.34
& 199.3
\\

FC-DFL
& 11.00
& $\mathbf{0.095}$
& $\mathbf{3.45}$
& $\mathbf{97.66}$
& 79.2
\\

Reference
& --
& 0.076
& 23.25
& 0.00
& 397.3
\\

\hline \hline
\end{tabular}
\end{table}

FC-DFL achieves the lowest mean IAFD of
0.09529 Hz$\cdot$h and decision regret of 97.66
among the deployable methods, as summarized in
Table~\ref{tab:main_performance}.
Compared with ERM-PTO, FC-DFL reduces IAFD
by 6.05\% and decision regret by 23.92\%.
A further 5.18\% IAFD reduction is achieved
relative to DS-DFL.
These improvements demonstrate the value of
training renewable forecasts according to the
realized cost of their induced BESS dispatch.

The benefit of FC-DFL is also evident
relative to Uncontrolled.
Its mean IAFD is reduced by 3.45\%,
whereas ERM-PTO and DS-DFL increase IAFD
by 2.77\% and 1.82\%, respectively.
ERM-PTO achieves the lowest forecasting
RMSE of 9.73 MW, but its dispatch produces
a higher IAFD than Uncontrolled.
This contrast highlights the importance of
forecast-induced BESS decisions in determining
realized regulation performance.
Moreover, FC-DFL reduces BESS throughput
by 61.17\% relative to ERM-PTO, achieving
better regulation with substantially less
charging and discharging activity.
Its mean IAFD is also 0.56\% below Scenario SD.

The IAFD distributions over 100 held-out
test windows are illustrated in
Fig.~\ref{fig:iafd_distribution}.
Each distribution is obtained by averaging
over five seeds within each test window.
Paired Wilcoxon signed-rank tests further
support the improvements achieved by FC-DFL
over ERM-PTO ($p=0.0066$) and DS-DFL
($p=0.0310$).
The corresponding mean IAFD reductions are
0.00614 and 0.00520 Hz$\cdot$h, respectively. Further analysis reveals larger improvements
in windows with high ERM-PTO decision regret.
After grouping the test windows into quartiles
according to \eqref{eq:decision_regret},
FC-DFL achieves mean IAFD reductions of
0.00352 and 0.02210 Hz$\cdot$h in the third
and fourth quartiles, respectively.
It outperforms ERM-PTO in 68\% and 72\%
of these windows.
These results indicate that the benefit of
decision-focused forecasting is particularly
pronounced when conventional forecasts
induce costly BESS operating decisions.

\begin{figure}[t]
\centering
\includegraphics[width=0.85\columnwidth]
{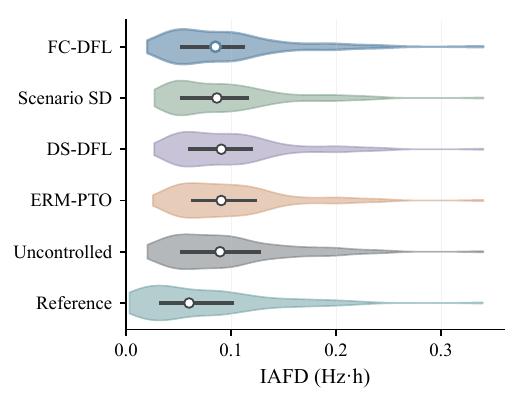}
\vspace{-1em}
\caption{Held-out test-window IAFD distributions.}
\label{fig:iafd_distribution}
\end{figure}

\subsection{Forecast-Induced BESS Operation}
\label{subsec:physical_consequences}

The physical effects of forecast-induced BESS
dispatch are illustrated in
Fig.~\ref{fig:physical_chain}, using seed 4
and test window 42 beginning on June 5, 2020,
at 22:00 UTC. The four panels trace the
relationship between renewable forecasts,
BESS operation, and realized frequency deviations. As illustrated in Fig.~\ref{fig:physical_chain}(a),
the renewable forecasts deviate from the realized
generation at different intervals, particularly
during substantial changes in renewable power.
These forecasts induce different BESS charging
and discharging schedules under the same physical
constraints.

\begin{figure*}[t]
\centering
\includegraphics[width=0.9\textwidth]
{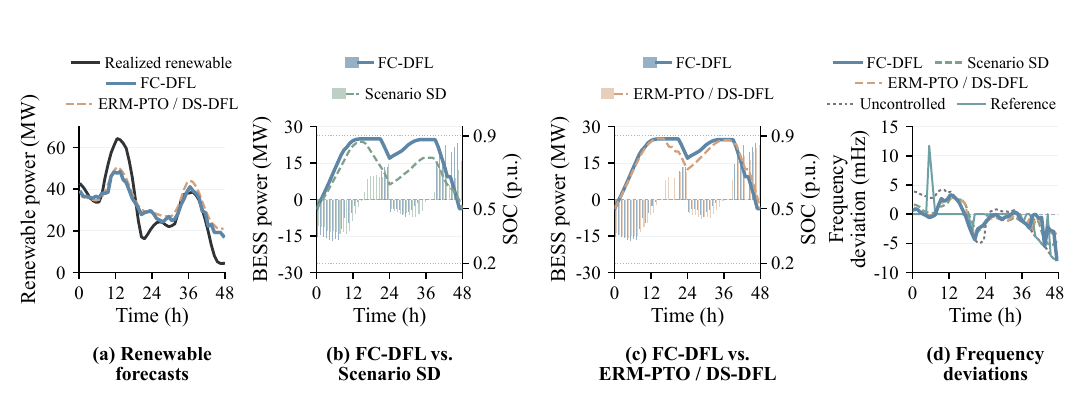}
\vspace{-1em}
\caption{Forecast-induced BESS operation in the representative test window: (a) renewable forecasts; (b) FC-DFL and Scenario SD BESS power and SOC; (c) FC-DFL and ERM-PTO/DS-DFL BESS power and SOC; and (d) realized frequency deviations.}
\label{fig:physical_chain}
\end{figure*}

The resulting BESS power and SOC trajectories
are presented in Figs.~\ref{fig:physical_chain}(b)
and (c). ERM-PTO and DS-DFL produce identical
dispatch trajectories in this window.
Compared with Scenario SD, FC-DFL maintains
a higher SOC over much of the middle horizon,
reflecting a different allocation of charging
and discharging power. These differences are
important because earlier BESS decisions change
the stored energy available for subsequent
regulation, as described by the SOC dynamics
in \eqref{eq:energy_dynamics}. The corresponding realized frequency deviations
are illustrated in Fig.~\ref{fig:physical_chain}(d).
FC-DFL achieves an IAFD of 0.07427 Hz$\cdot$h,
representing a 47.37\% reduction relative to
Uncontrolled (0.14112 Hz$\cdot$h).
Further IAFD reductions of 4.69\% and 8.57\%
are obtained relative to ERM-PTO/DS-DFL
and Scenario SD, respectively.
Meanwhile, FC-DFL requires 388.87 MWh of BESS
throughput, compared with 437.51 MWh for
ERM-PTO/DS-DFL.
The results illustrate how forecast-induced
BESS dispatch and intertemporal SOC evolution
affect realized regulation performance.


\subsection{Robustness Across BESS and Renewable Conditions}
\label{subsec:robustness}

The robustness of FC-DFL is evaluated under different
BESS operating conditions. The energy capacity is varied
from 200 to 600 MWh, and the ramp-rate limit from
30 to 90 MW/h. The learned predictors remain fixed,
while BESS dispatch is reoptimized for each configuration.
Across all tested configurations, FC-DFL achieves the
lowest mean IAFD among the deployable methods, as
illustrated in Figs.~\ref{fig:robustness}(a) and (b).
Relative to ERM-PTO, the IAFD reduction remains
between 6.05\% and 6.10\% for different energy
capacities and between 5.40\% and 6.13\% for
different ramp-rate limits. These results demonstrate
that the improvement is maintained across the tested
BESS operating conditions.

\begin{figure}[t]
\centering
\includegraphics[width=0.85\columnwidth]
{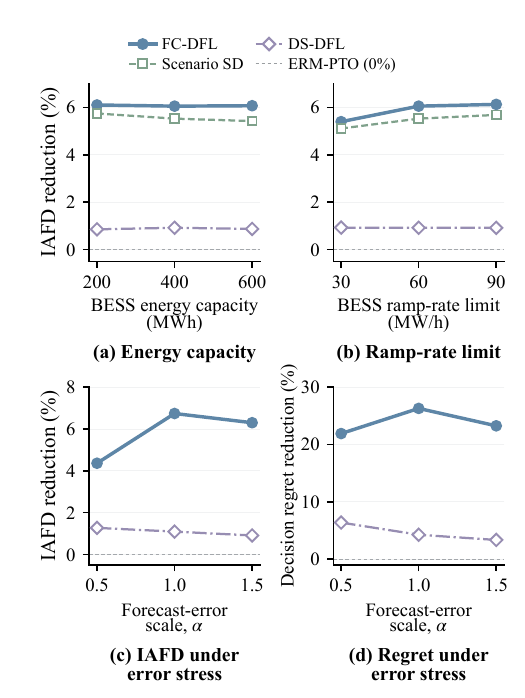}
\vspace{-1em}
\caption{Robustness to BESS operating limits and
renewable forecast-error stress, which reductions are
relative to ERM-PTO.}
\label{fig:robustness}
\end{figure}

Renewable forecast-error stress is further examined
using constructed renewable realizations based on
ERM-PTO forecasts and empirical validation-error
trajectories scaled by $\alpha\in\{0.5,1.0,1.5\}$.
For each seed and test window, the same error
trajectory is scaled across the three stress levels.
The learned predictors and exogenous balance
requirement remain unchanged. Under these conditions, FC-DFL consistently
achieves lower IAFD and decision regret than
ERM-PTO and DS-DFL, as presented in
Figs.~\ref{fig:robustness}(c) and (d).
The IAFD reductions relative to ERM-PTO reach
4.36\%, 6.74\%, and 6.30\% at the three stress
levels, respectively. Corresponding reductions
in decision regret range from 21.9\% to 26.3\%.
Together, these results demonstrate that FC-DFL
maintains its regulation advantage across the
tested BESS configurations and renewable
forecast-error conditions.

\subsection{Decision-Oriented Importance and Operational Faithfulness}
\label{subsec:xai_case}

The XAI evaluation considers 50 selected seed-window
cases, including 45 with nonzero decision feedback.
The antithetic estimator uses $K_{\mathrm{xai}}=64$,
with prediction sensitivities computed for
standardized inputs. The resulting feature-time
importance scores in \eqref{eq:estimated_xai_score}
are illustrated in Fig.~\ref{fig:xai_faithfulness}(a).
Among the 98 predictor inputs, the 1-h and 2-h
renewable-history lags receive the highest aggregate
importance, followed by the day-of-year features.
The importance also varies across forecast periods,
indicating that the influence of an input on BESS
operation depends on the timing of its forecast contribution.

\begin{figure}[t]
\centering
\includegraphics[width=0.85\columnwidth]
{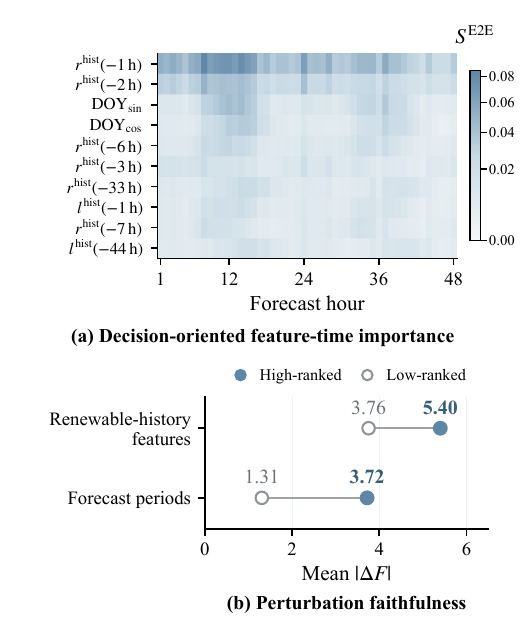}
\vspace{-1em}
\caption{Decision-oriented importance and operational faithfulness: (a) feature-time importance and (b) matched-perturbation faithfulness.}
\label{fig:xai_faithfulness}
\end{figure}

A clear distinction emerges between prediction
sensitivity and decision-oriented importance.
For temporal comparison, prediction sensitivity
is aggregated across input features as
$\sum_j G_{i,j,t}^{\mathrm{pred}}$ and compared
with the estimated temporal decision importance
$\widehat{\bar g}_{i,t}^{\mathrm{opt}}$.
As reported in Table~\ref{tab:xai_faithfulness},
the mean Spearman correlation between their
temporal rankings is only 0.019, with an average
top-5 overlap of 0.67 periods.
The five highest-ranked decision periods account
for 28.7\% of the total decision weight, compared
with 10.4\% under uniform weighting.
These results highlight the difference between
forecast periods sensitive to predictor inputs
and those important to realized BESS operation. Operational faithfulness is further examined
through matched perturbations and exact BESS
reoptimization, following
\eqref{eq:feature_faithfulness} and
\eqref{eq:temporal_faithfulness}.
High- and low-ranked renewable-history groups
each contain five eligible features, perturbed
by 0.1 training standard deviation.
Forecast-period groups contain five periods
with matched feasible perturbations of up to 1 MW.
The resulting changes in realized decision cost
are evaluated under the same renewable realization.


\begin{table}[t]
\centering
\caption{Decision-Oriented XAI Evaluation}
\label{tab:xai_faithfulness}
\renewcommand{\arraystretch}{1.12}
\setlength{\tabcolsep}{2.8pt}

\begin{tabular}{lcccc}
\hline\hline

\multicolumn{5}{l}{
\parbox{0.98\columnwidth}{
\textit{A. Importance of different period}
}} \\
\hline
\multicolumn{4}{l}{Mean Spearman rank correlation}
& 0.019 \\
\multicolumn{4}{l}{Mean top-5 overlap}
& 0.67 \\
\multicolumn{4}{l}{Decision weight captured by top-5 periods}
& 28.7\% \\

\hline
\multicolumn{5}{l}{
\parbox{0.98\columnwidth}{
\textit{B. Operational Validation Using Matched Perturbations}
}} \\
\multicolumn{5}{l}{
\textit{Mean absolute change in realized operating cost, $|\Delta F|$}
} \\
\hline

Group
& {High}
& {Low}
& Ratio
& $p$-value \\
\hline

Renewable history
& 5.40
& 3.76
& 1.44
& 0.0147 \\

Forecast periods
& 3.72
& 1.31
& $\mathbf{2.84}$
& $1.76\times10^{-4}$ \\

\hline\hline
\end{tabular}
\end{table}

The operational relevance of the importance rankings
is evaluated in Fig.~\ref{fig:xai_faithfulness}(b)
and Table~\ref{tab:xai_faithfulness}.
High-ranked renewable-history groups produce
a mean $|\Delta F|$ of 5.40, compared with
3.76 for low-ranked groups, representing
a 1.44-fold difference.
The distinction is more pronounced for
forecast periods, where high-ranked groups
produce 2.84 times the realized cost change
of low-ranked groups (3.72 versus 1.31).
Both differences are statistically significant
under paired Wilcoxon signed-rank tests. The larger realized cost changes for high-ranked
groups support the operational relevance of
the proposed importance rankings.
Unlike prediction sensitivity alone, the
decision-oriented importance scores account for
the sensitivity of realized BESS cost to forecast
changes. The matched perturbation results further
demonstrate that the decision feedback used for
FC-DFL training also provides a basis for
interpreting the operating consequences of
renewable forecasts.

\subsection{Sensitivity to Smoothing Radius and Perturbation Budget}
\label{subsec:sigma_k}

The Gaussian smoothing radius $\sigma$ and
perturbation budget $K$ govern the estimation
of decision feedback in \eqref{eq:one_sided_estimator}.
The smoothing radius determines the scale of
renewable forecast perturbations used to evaluate
realized BESS cost, while $K$ controls the number
of perturbation directions used for each gradient
estimate. Their effects are evaluated by varying
one parameter at a time around the nominal setting
$(\sigma_0,K)=(1.3483\ \mathrm{MW},4)$. Unlike the BESS operating-limit tests in
Section~\ref{subsec:robustness}, these experiments
retain the same physical BESS constraints
and realized-cost evaluation.
They examine the influence of gradient-estimation
settings on the decision feedback used for
renewable predictor training.

\begin{figure}[t]
\centering
\includegraphics[width=0.8\columnwidth]
{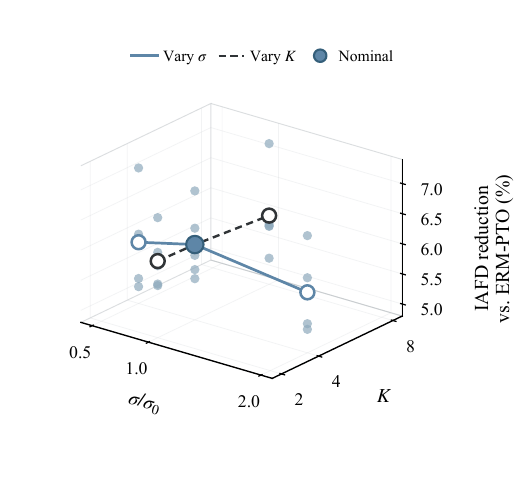}
\vspace{-3em}
\caption{Sensitivity of FC-DFL to the Gaussian smoothing radius and perturbation budget.}
\label{fig:sigma_k_sensitivity}
\end{figure}

With $K=4$, the mean IAFD remains between
0.09529 and 0.09554 Hz$\cdot$h as
$\sigma/\sigma_0$ varies from 0.5 to 2.0.
The nominal smoothing radius achieves the
lowest decision regret of 97.66 among these
three settings, compared with 98.84 and 98.90
at the smaller and larger radii, respectively.
The small differences indicate that the
realized regulation performance is relatively
insensitive to the smoothing radius within
the examined range. The influence of the perturbation budget is
examined at $\sigma=\sigma_0$.
Increasing $K$ from 2 to 8 does not produce
a consistent improvement in regulation performance.
As reported in Table~\ref{tab:sigma_k},
the mean IAFD varies from 0.09527 to
0.09540 Hz$\cdot$h, while the offline
training time increases from 135.82 to
332.76 s. Thus, a larger perturbation budget
increases computational effort without a
corresponding reduction in realized IAFD
under the tested settings.

\begin{table}[t]
\centering
\caption{Sensitivity of FC-DFL to Smoothing Radius and Perturbation Budget}
\label{tab:sigma_k}
\renewcommand{\arraystretch}{1.12}
\setlength{\tabcolsep}{5pt}
\begin{tabular}{ccccc}
\hline \hline
$\sigma/\sigma_0$
& $K$
& IAFD (Hz$\cdot$h)
& Decision regret 
& Training time \\
\hline
0.5 & 4 & 0.0955 & 98.84 & 243.10 s \\
1.0 & 2 & 0.0953 & 97.56 & 135.82 s \\
1.0 & 4 & 0.0953 & 97.66 & 195.79 s \\
1.0 & 8 & 0.0954 & 98.22 & 332.76 s \\
2.0 & 4 & 0.0955 & 98.90 & 126.72 s \\
\hline \hline
\end{tabular}
\end{table}

Across all five configurations, FC-DFL maintains
an IAFD reduction of 5.80\%--6.07\% relative
to ERM-PTO, as illustrated in
Fig.~\ref{fig:sigma_k_sensitivity}.
These results demonstrate that the regulation
benefits of FC-DFL are maintained under
different smoothing radii and perturbation
budgets, supporting the numerical robustness
of the proposed decision-feedback training.


\section{Conclusion}
\label{sec:conclusion}

This paper developed Frequency-Calibrated Decision-Focused Learning (FC-DFL) to align renewable forecasting with the operational objectives of constrained BESS dispatch. FC-DFL leverages feedback from the realized outcomes of forecast-driven BESS dispatch to guide predictor training, enabling forecasts to be optimized for their downstream operational results. Building on the same decision feedback, we further developed a decision-oriented XAI method that combines it with prediction sensitivity to identify the forecasting inputs and periods most relevant to BESS dispatch.
Case studies using German wind and solar data demonstrated that FC-DFL reduced the integrated absolute frequency-deviation metric by 6.05\% and decision regret by 23.92\% compared with accuracy-based predict-then-optimize approach. BESS energy throughput was also reduced by 61.17\%. The XAI analysis further revealed that high-ranked forecast periods produced mean absolute realized decision-cost changes 2.84 times those of low-ranked periods, supporting the operational relevance of the proposed importance rankings.
These findings highlight the importance of learning from downstream decision-making results in renewable forecasting. The proposed method can support both predictor training and identify the forecasting information that matters more to BESS operation.  Future work will extend FC-DFL to network-constrained BESS operation and dynamic frequency-response models.

\bibliographystyle{IEEEtran}
\bibliography{references}

\end{document}